\documentclass[11pt]{article}
\usepackage{amsmath,amsthm,amssymb}
\usepackage{stmaryrd}
\usepackage{algpseudocode,algorithm}
\usepackage{fullpage}
\usepackage[english]{babel}
\usepackage{latexsym}
\usepackage{verbatim}
\usepackage{enumitem}
\usepackage{color}
\usepackage{mathtools}
\usepackage{cleveref}

\makeatletter
\newcommand{\pushright}[1]{\ifmeasuring@#1\else\omit\hfill$\displaystyle#1$\fi\ignorespaces}
\newcommand{\pushleft}[1]{\ifmeasuring@#1\else\omit$\displaystyle#1$\hfill\fi\ignorespaces}
\makeatother

\newcommand{\proj}{\ensuremath{e}}
\newcommand{\e}[2]{\ensuremath{\proj_{#1}^{(#2)}}}
\DeclareMathOperator{\argmin}{argmin}
\DeclareMathOperator{\wops}{{\bf W}}
\DeclareMathOperator{\ops}{{\mathcal{O}}}

\newcommand{\insta}{\ensuremath{I}}

\DeclareMathOperator{\rel}{{\bf R}}
\newcommand{\tup}[1]{\ensuremath{\mathbf #1}}
\newcommand{\tuple}[1]{\ensuremath{(#1)}}
\newcommand{\ignore}[1]{}

\DeclareMathOperator{\wrel}{{\bf \Phi}}
\DeclareMathOperator{\supp}{{\rm supp}}
\DeclareMathOperator{\VCSP}{VCSP}
\DeclareMathOperator{\wPol}{wPol}
\DeclareMathOperator{\Feas}{Feas}
\DeclareMathOperator{\Opt}{Opt}

\DeclareMathOperator{\dom}{{Feas}}

\DeclareMathOperator{\wClone}{wClone}
\DeclareMathOperator{\Imp}{Imp}
\DeclareMathOperator{\Inv}{Inv}

\DeclareMathOperator{\Pol}{Pol}
\DeclareMathOperator{\wRelClone}{wRelClone}

\newcommand{\qq}{\ensuremath{\overline{\mathbb{Q}}}}
\newcommand{\q}{\ensuremath{\mathbb{Q}}}
\newcommand{\rr}{\ensuremath{\overline{\mathbb{R}}}}
\renewcommand{\r}{\ensuremath{\mathbb{R}}}

\newcommand{\eps}{\chi}
\newcommand{\JD}{\ensuremath{\mathbf{J}_D}}
\newcommand{\costeq}[1]{\ensuremath{#1_{\sim}}}

\newcommand{\Rwrel}{\ensuremath{\wrel^\r}}
\newcommand{\Rwops}{\ensuremath{\wops^\r}}
\newcommand{\RImp}{\ensuremath{\Imp_\r}}
\newcommand{\RwPol}{\ensuremath{\wPol_\r}}
\newcommand{\RwRelClone}{\ensuremath{\wRelClone_\r}}
\newcommand{\RwClone}{\ensuremath{\wClone_\r}}

\newcommand{\inn}[2]{\ensuremath{\langle#1,#2\rangle}}

\theoremstyle{plain}
\newtheorem{theorem}{Theorem}
\newtheorem{corollary}[theorem]{Corollary}
\newtheorem{lemma}[theorem]{Lemma}

\theoremstyle{definition}

\newtheorem{definition}[theorem]{Definition}
\newtheorem{example}[theorem]{Example}

\newtheorem*{theorem*}{Theorem}

\newlist{props}{enumerate}{1}
\setlist[props]{label={\rm (\roman*)}}

\algtext*{EndWhile}
\algtext*{EndIf}
\algtext*{EndFor}

\ignore{

\newcommand{\ignore}[1]{}

\DeclareMathOperator*{\argmin}{arg\,min}

\newcommand{\tuple}[1]{\ensuremath{\langle #1 \rangle}}

\newcommand{\VCSP}[1]{{VCSP}(\ensuremath{#1})}

}

\begin{document}

\title{\Large A Galois Connection for Weighted (Relational) Clones of Infinite
Size\thanks{The authors were supported by a Royal Society Research Grant.
Stanislav \v{Z}ivn\'y was supported by a Royal Society University Research
Fellowship. A part of this work appeared in \emph{Proceedings of the 42nd International Colloquium on Automata, Languages, and
Programming (ICALP)}, 2015~\cite{fz15:icalp}.}}

\author{Peter Fulla\\
University of Oxford, UK\\
\texttt{peter.fulla@cs.ox.ac.uk}
\and
Stanislav \v{Z}ivn\'y\\
University of Oxford, UK\\
\texttt{standa.zivny@cs.ox.ac.uk}
}

\date{}

\maketitle

\begin{abstract}
A Galois connection between clones and relational clones on a fixed finite
domain is one of the cornerstones of the so-called algebraic approach to the
computational complexity of non-uniform Constraint Satisfaction Problems (CSPs).
Cohen et al. established a Galois connection between \emph{finitely-generated}
weighted clones and \emph{finitely-generated} weighted relational
clones~[SICOMP'13], and asked whether this connection holds in general. We
answer this question in the affirmative for weighted (relational) clones with
\emph{real} weights and show that the complexity of the corresponding valued
CSPs is preserved.
\end{abstract}

\section{Introduction}\label{sec:intro}

The constraint satisfaction problem (CSP) is a general framework capturing
decision problems arising in many contexts of computer
science~\cite{Hell08:survey}. The CSP is NP-hard in general but there has been
much success in finding tractable fragments of the CSP by restricting the types
of relations allowed in the constraints.  A set of allowed relations has been
called a \emph{constraint language}~\cite{Feder98:monotone}. For some constraint
languages, the associated constraint satisfaction problems with constraints
chosen from that language are solvable in polynomial-time, whilst for other
constraint languages this class of problems is NP-hard~\cite{Feder98:monotone};
these are referred to as \emph{tractable languages} and \emph{NP-hard
languages}, respectively. Dichotomy theorems, which classify each possible
constraint language as either tractable or NP-hard, have been established for
constraint languages over two-element domains~\cite{Schaefer78:complexity},
three-element domains~\cite{Bulatov06:3-elementJACM}, for conservative
(containing all unary relations) constraint
languages~\cite{Bulatov11:conservative}, for maximal constraint
languages~\cite{Bulatov01:complexity,Bulatov04:graph}, for graphs (corresponding
to languages containing a single binary symmetric
relation)~\cite{Hell90:h-coloring}, and for digraphs
(corresponding to languages containing a single binary relation)
without sources and sinks~\cite{Barto09:siam}. The most successful approach to classifying the
complexity of constraint languages has been the algebraic
approach~\cite{Jeavons97:closure,Bulatov05:classifying,Barto14:jacm}. The
dichotomy conjecture of Feder and Vardi~\cite{Feder98:monotone} asserts that every
constraint language is either tractable or NP-hard, and the algebraic refinement
of the conjecture specifies the precise boundary between tractable and NP-hard
languages~\cite{Bulatov05:classifying}.

The \emph{valued} constraint satisfaction problem (VCSP) is a general framework
that captures not only feasibility problems but also optimisation
problems~\cite{Cohen06:complexitysoft,z12:complexity,jkz14:beatcs}. A VCSP
instance represents each constraint by a \emph{weighted relation}, which is a
$\qq$-valued function where $\qq=\mathbb{Q}\cup\{\infty\}$, and the goal is to
find a labelling of variables minimising the sum of the values assigned by the
constraints to that labelling.
Tractable fragments of the VCSP have
been identified by restricting the types of allowed weighted relations that can
be used to define the valued constraints. A set of allowed weighted relations
has been called a \emph{valued constraint
language}~\cite{Cohen06:complexitysoft}. 
Dichotomy theorems, which
classify each possible valued constraint language as either tractable or
NP-hard, have been established for valued constraint languages over two-element
domains~\cite{Cohen06:complexitysoft}, for conservative (containing all
$\{0,1\}$-valued unary cost functions) valued constraint
languages~\cite{kz13:jacm}, for finite-valued (all weighted relations
are $\q$-valued) constraint languages~\cite{tz13:stoc}. Moreover, it has been
shown that a dichotomy for constraint languages implies a dichotomy for
valued constraint languages~\cite{Kolmogorov15:general-valued}.
Finally, the power of the basic linear programming
relaxation~\cite{tz12:focs,ktz15:sicomp} and the power of the Sherali-Adams
relaxation~\cite{tz15:icalp,tz15:power} for valued constraint languages have been
completely characterised.

Cohen et al. have introduced an algebraic theory of weighted
clones~\cite{cccjz13:sicomp}, further extended in~\cite{tz15:sidma,Kozik15:icalp}, for classifying the computational complexity of
valued constraint languages. This theory establishes a one-to-one correspondence
between valued constraint languages closed under expressibility (which does not
change the complexity of the associated class of optimisation problems), called
weighted relational clones, and weighted clones~\cite{cccjz13:sicomp}. This is
an extension of (a part of) the algebraic approach to CSPs which relies on a
one-to-one correspondence between constraint languages closed under
pp-definability (which does not change the complexity of the associated class of
decision problems), called relational clones, and
clones~\cite{Bulatov05:classifying}, thus making it possible to use deep results
from universal algebra. This theory has been developed primarily as an aid for
studying the computational complexity of valued CSPS (and indeed, recent
progress on valued
CSPs~\cite{tz12:focs,tz13:stoc,ktz15:sicomp,tz15:icalp,Kolmogorov15:general-valued}
and on special cases of valued CSPs~\cite{Uppman13:icalp} heavily rely on the
theory introduced in~\cite{cccjz13:sicomp}), but the theory is interesting in its own
right~\cite{tz15:sidma,Kozik15:icalp,Vancura14:msc,Vaicenavicius14:msc}.

\subsection*{Contributions}
The Galois connection between weighted clones and weighted relational clones
established in~\cite{cccjz13:sicomp} was proved only for weighted (relational)
clones generated by a \emph{finite} set. The authors asked whether
such a correspondence holds also for weighted (relational) clones in general. In
this paper we answer this question in the affirmative.

Firstly, we show that the Galois connection from~\cite{cccjz13:sicomp} (using
only rational weights) does \emph{not} work for general weighted (relational)
clones. Secondly, we alter the definition of weighted (relational) clones and
establish a new Galois connection that holds even when the generating set has an
infinite size. We allow weighted relations and weightings to assign real weights
instead of rational, require weighted relational clones to be closed under
operator $\Opt$ and to be topologically closed, and prove that these changes
preserve tractability of a constraint language (albeit only in an approximate
sense).

Including the $\Opt$ operator (sometimes called $\argmin$)\footnote{Given a
$k$-ary function $f:D^k\to\qq$, $\Opt(f)$ is the $k$-ary relation over $D$ of
minimal-value tuples of $f$.}
in the definition of weighted relational clones
simplifies the structure of the space of all weighted clones, and guarantees
that every non-projection polymorphism of a weighted relational clone $\Gamma$
is assigned a positive weight by some weighted polymorphism of $\Gamma$. Indeed,
including the $\Opt$ operator is very natural and can be used to simplify several
results in~\cite{Kozik15:icalp}.
Real weights (as opposed to rational weights)
were previously used, in the context of valued CSPs, in~\cite{tz13:stoc} and our
results confirm that real weights are necessary when studying infinite weighted
(relational) clones.

The proof of the Galois connection in~\cite{cccjz13:sicomp} relies on results on
linear programming duality; we use their generalisation from the theory of
convex optimisation in order to establish the connection even for infinite sets.

\section{Background}
\label{sec:prelim}

\subsection{Valued CSPs}

Throughout the paper, let $D$ be a fixed finite set of size at least two.

\begin{definition}\label{def:rel}
An $m$-ary \emph{relation}%
\footnote{An $m$-ary relation $R$ over $D$ is commonly defined as a subset of
$D^m$. For the corresponding mapping $\phi : D^m \to \{0, \infty\}$ it holds
$\phi(\tup{x}) = 0$ when $\tup{x} \in R$ and $\phi(\tup{x}) = \infty$ otherwise.
We shall use both definitions interchangeably. Because two mappings that differ
only by a constant are usually equivalent for our purposes, we consider mappings
$D^m \to \{c, \infty\}$ to be relations even if $c \neq 0$.
}
over $D$ is any mapping $\phi:D^m\to\{c,\infty\}$ for
some $c\in\mathbb{Q}$.
We denote by $\rel_D^{(m)}$ the set of all $m$-ary relations and let
$\rel_D=\bigcup_{m\geq 1}{\rel_D^{(m)}}$.
\end{definition}

Given an $m$-tuple $\tup{x}\in D^m$, we denote its $i$th entry by $\tup{x}[i]$ for $1\leq i\leq m$.

Let $\qq=\mathbb{Q}\cup\{\infty\}$ denote the set of rational numbers with (positive) infinity.

\begin{definition}\label{def:wrel}
An $m$-ary \emph{weighted relation} over $D$ is any mapping $\gamma:D^m\to\qq$.
We denote by $\wrel_D^{(m)}$ the set of all $m$-ary weighted relations and let
$\wrel_D=\bigcup_{m\ge 1}{\wrel_D^{(m)}}$.
\end{definition}

From Definition~\ref{def:wrel} we have that relations are a special type of
weighted relations.

\begin{example}
An important example of a (weighted) relation is the binary equality $\phi_=$ on
$D$ defined by
$\phi_=(x,y)=0$ if $x=y$ and $\phi_=(x,y)=\infty$ if $x\neq y$.

Another example of a relation is the unary empty relation $\phi_\emptyset$
defined on $D$ by $\phi_\emptyset(x)=\infty$ for all $x\in D$.
\end{example}

For any $m$-ary weighted relation $\gamma\in\wrel_D^{(m)}$, we denote by
$\Feas(\gamma)=\{\tup{x}\in D^m\:|\:\gamma(\tup{x})<\infty\}\in\rel_D^{(m)}$ the
underlying \emph{feasibility relation}, and by $\Opt(\gamma)=\{\tup{x}\in
\Feas(\gamma)\:|\:\gamma(\tup{x})\leq \gamma(\tup{y}) \mbox{ for every }\tup{y}\in
D^m\} \in\rel_D^{(m)}$ the relation of minimal-value tuples.

\begin{definition}
Let $V=\{x_1,\ldots, x_n\}$ be a set of variables. A \emph{valued constraint} over $V$ is an expression
of the form $\gamma(\tup{x})$ where $\gamma\in \wrel_D^{(m)}$ and $\tup{x}\in V^m$. The number $m$ is called the \emph{arity} of the constraint,
the weighted relation $\gamma$ is called the \emph{constraint weighted relation},
and the tuple $\tup{x}$ the \emph{scope} of the constraint.
\end{definition}

We call $D$ the \emph{domain}, the elements of $D$ \emph{labels} (for variables), and say that the
weighted relations in $\wrel_D$ take \emph{values} or \emph{weights}.

\begin{definition}
An instance of the \emph{valued constraint satisfaction problem} (VCSP) is specified
by a finite set $V=\{x_1,\ldots,x_n\}$ of variables, a finite set $D$ of labels,
and an \emph{objective function} $\insta$
expressed as follows:
\begin{equation}
\insta(x_1,\ldots, x_n)=\sum_{i=1}^q{\gamma_i(\tup{x}_i)}\,,
\label{eq:sepfun}
\end{equation}
where each $\gamma_i(\tup{x}_i)$, $1\le i\le q$, is a valued constraint over $V$.
Each constraint can appear multiple times in $\insta$.

The goal is to find an \emph{assignment} (or a \emph{labelling}) of labels to the variables that minimises $\insta$.
\end{definition}

CSPs are a special case of VCSPs using only (unweighted) relations with the goal
to determine the existence of a feasible assignment.

\begin{definition}
Any set $\Gamma\subseteq\wrel_D$ is called a \emph{(valued) constraint
language} over $D$, or simply a \emph{language}. We
will denote by $\VCSP(\Gamma)$ the class of all VCSP instances in which the
constraint weighted relations are all contained in $\Gamma$.
\end{definition}

\begin{definition}\label{defn:tract}
A constraint language $\Gamma$ is called \emph{tractable} if
$\VCSP(\Gamma')$ can be solved (to optimality) in
polynomial time for every finite subset $\Gamma'\subseteq\Gamma$, and $\Gamma$
is called \emph{intractable} if $\VCSP(\Gamma')$ is NP-hard for some finite
$\Gamma'\subseteq\Gamma$.
\end{definition}

We refer the reader to a recent survey~\cite{jkz14:beatcs} for more information
on the computational complexity of constraint languages.

\subsection{Weighted relational clones}

\begin{definition}
A weighted relation $\gamma$ of arity $r$ can be obtained by \emph{addition}
from the weighted relation $\gamma_1$ of arity $s$ and the weighted relation
$\gamma_2$ of arity $t$ if $\gamma$ satisfies the identity
\begin{equation}
\gamma(x_1,\dots,x_r) = \gamma_1(y_1,\dots,y_s) + \gamma_2(z_1,\dots,z_t)
\end{equation}
for some (fixed) choice of $y_1,\dots,y_s$ and $z_1,\dots,z_t$ from amongst
$x_1,\dots,x_r$.
\end{definition}

\begin{definition}
A weighted relation $\gamma$ of arity $r$ can be obtained by \emph{minimisation}
from the weighted relation $\gamma'$ of arity $r+s$ if $\gamma$ satisfies the
identity
\begin{equation}
\gamma(x_1,\dots,x_r) =
\min_{(y_1,\dots,y_s)\in D^s} \gamma'(x_1,\dots,x_r,y_1,\dots,y_s) \,.
\end{equation}
\end{definition}

\begin{definition}
\label{defn:wrelclone}
A constraint language $\Gamma \subseteq \wrel_D$ is called a \emph{weighted
relational clone} if it contains the binary equality relation $\phi_=$ and the
unary empty relation $\phi_\emptyset$,%
\footnote{Although the definition in~\cite{cccjz13:sicomp} does not require
the inclusion of $\phi_\emptyset$, the proofs there implicitly assume its presence
in any weighted relational clone.}
and is closed under addition,
minimisation, scaling by non-negative rational constants, and addition of
rational constants.

For any $\Gamma$, we define $\wRelClone(\Gamma)$ to be the smallest weighted relational clone containing $\Gamma$.
\end{definition}

Note that for any weighted relational clone $\Gamma$, if $\gamma\in\Gamma$ then
$\Feas(\gamma)\in\Gamma$ as $\Feas(\gamma)=0\gamma$ (we define $0\cdot\infty =
\infty$).

\begin{definition}
\label{defn:express}
Let $\Gamma \subseteq \wrel_D$ be a constraint language, $I \in \VCSP(\Gamma)$
an instance with variables $V$, and $L = \tuple{v_1, \dots, v_r}$ a list of
variables from $V$. The \emph{projection} of $I$ onto $L$, denoted $\pi_L(I)$,
is the $r$-ary weighted relation on $D$ defined as
\begin{equation}
\pi_L(I)(x_1,\dots,x_r) =
\min_{\{s:V\to D ~|~ \tuple{s(v_1),\dots,s(v_r)}=\tuple{x_1,\dots,x_r}\}}
I(s) \,.
\end{equation}
We say that a weighted relation $\gamma$ is \emph{expressible} over a constraint
language $\Gamma$ if $\gamma = \pi_L(I)$ for some $I \in \VCSP(\Gamma)$ and list
of variables $L$. We call the pair $\tuple{I,L}$ a \emph{gadget} for expressing
$\gamma$ over $\Gamma$.
\end{definition}

The list of variables $L$ in a gadget may contain repeated entries. The minimum
over an empty set is $\infty$.
\begin{example}
\label{ex:weightedequality}
For any $\Gamma \subseteq \wrel_D$, we can express the binary equality relation
$\phi_=$ on $D$ over language $\Gamma$ using the following gadget. Let $I \in
\VCSP(\Gamma)$ be the instance with a single variable $v$ and no constraints,
and let $L = \tuple{v,v}$. Then, by \Cref{defn:express}, $\pi_L(I) = \phi_=$.
\end{example}

We may equivalently define a weighted relational clone as a set $\Gamma
\subseteq \wrel_D$ that contains the unary empty relation $\phi_\emptyset$ and
is closed under expressibility, scaling by non-negative rational constants, and
addition of rational constants~\cite[Proposition 4.5]{cccjz13:sicomp}.

The following result has been shown in~\cite{cccjz13:sicomp}.

\begin{theorem}\label{thm:relclonetract}
A constraint language $\Gamma$ is tractable if and only if $\wRelClone(\Gamma)$
is tractable, and $\Gamma$ is intractable if and only if $\wRelClone(\Gamma)$ is
intractable.
\end{theorem}

Consequently, when trying to identify tractable constraint languages, it is
sufficient to consider only weighted relational clones.

\subsection{Weighted clones}

Any mapping $f:D^k\rightarrow D$ is called a $k$-ary \emph{operation}. We will
apply a $k$-ary operation $f$ to $k$ $m$-tuples $\tup{x_1},\ldots,\tup{x_k}\in
D^m$ coordinatewise, that is,
\begin{equation}
f(\tup{x_1},\ldots,\tup{x_k})=(f(\tup{x_1}[1],\ldots,\tup{x_k}[1]),\ldots,f(\tup{x_1}[m],\ldots,\tup{x_k}[m]))\in D^m\,.
\end{equation}
\begin{definition} \label{def:pol}
Let $\gamma$ be an $m$-ary weighted relation on $D$ and let $f$ be a $k$-ary operation
on $D$.
Then $f$ is a \emph{polymorphism} of $\gamma$ if,
for any $(\tup{x_1},\ldots,\tup{x_k}) \in (\Feas(\gamma))^k$,
we have $f(\tup{x_1},\ldots,\tup{x_k})\in\dom(\gamma)$.

For any constraint language $\Gamma$ over a set $D$,
we denote by $\Pol(\Gamma)$ the set of all operations on $D$ which are polymorphisms of all
$\gamma \in \Gamma$. We write $\Pol(\gamma)$ for $\Pol(\{\gamma\})$.
\end{definition}
A $k$-ary \emph{projection} is an operation of the form
$\proj^{(k)}_i(x_1,\ldots,x_k)=x_i$ for some $1\leq i\leq k$.
Projections are (trivial) polymorphisms of all constraint languages.

\begin{definition}
The \emph{superposition} of a $k$-ary operation $f:D^k\rightarrow D$ with $k$
$\ell$-ary operations $g_i:D^\ell\rightarrow D$ for $1\leq i\leq k$ is the
$\ell$-ary function $f[g_1,\ldots,g_k]:D^\ell\to D$ defined by
\begin{equation}
f[g_1,\ldots,g_k](x_1,\ldots,x_\ell)=f(g_1(x_1,\ldots,x_\ell),\ldots,g_k(x_1,\ldots,x_\ell))\,.
\end{equation}
\end{definition}

\begin{definition}\label{def:clone}
A \emph{clone} of operations, $C$, is a set of operations on $D$
that contains all projections and is closed under superposition.
The $k$-ary operations in a clone $C$ will be denoted by $C^{(k)}$.
\end{definition}

\begin{example}\label{ex:clone}
For any $D$, let $\JD$ be the set of all projections on $D$. By
Definition~\ref{def:clone}, $\JD$ is a clone.
\end{example}

It is well known that $\Pol(\Gamma)$ is a clone for all constraint
languages $\Gamma$.

\begin{definition} \label{defn:wop}
A $k$-ary \emph{weighting}
of a clone $C$ is a function $\omega : C^{(k)} \rightarrow \mathbb{Q}$ such that
$\omega(f) < 0$ only if $f$ is a projection and
\begin{equation}\label{eqZeroWeightsSum}
  \sum_{f \in C^{(k)}}\omega(f)\ =\ 0\,.
\end{equation}
We will call a function $\omega : C^{(k)} \rightarrow \mathbb{Q}$ that satisfies
\Cref{eqZeroWeightsSum} but assigns a negative weight to some operation $f
\not\in \JD^{(k)}$ an \emph{improper weighting}. In order to emphasise the
distinction we may also call a weighting a \emph{proper weighting}.

When specifying a weighting, we often write it as a weighted sum of operations
(i.e.\ $\sum \omega(f)\cdot f$) without any zero terms.
\end{definition}
\begin{definition} \label{defn:wp_trans}
For any clone $C$, any $k$-ary weighting $\omega$ of $C$, and any $g_1,\ldots,g_k \in
C^{(\ell)}$, the \emph{superposition}
of $\omega$
and $g_1,\ldots,g_k$ is the function $\omega[g_1,\ldots,g_k]: C^{(\ell)}
\rightarrow \mathbb{Q}$ defined by
\begin{equation}
\omega[g_1,\ldots,g_k](f') = \sum_{f \in C^{(k)}\:\wedge\:f[g_1,\ldots,g_k] =
f'}\omega(f)\,.
\end{equation}
By convention, the value of an empty sum is $0$.

If the result of a superposition is a proper weighting (that is, negative weights
are only assigned to projections), then that superposition will be called a
\emph{proper} superposition.
\end{definition}
\begin{definition} \label{defn:wclone}
A \emph{weighted clone}, $\Omega$, is a non-empty set of weightings of some fixed
clone $C$, called the \emph{support clone} of $\Omega$, which is closed under scaling
by non-negative rational constants, addition of weightings of equal arity, and proper superposition with
operations from $C$.
\end{definition}

We now link weightings and weighted relations by the concept of weighted
polymorphism, which will allow us to establish a correspondence between weighted
clones and weighted relational clones.
\begin{definition} \label{def:wp}
Let $\gamma$ be an $m$-ary weighted relation
on $D$ and let $\omega$ be a $k$-ary weighting of a clone $C$ of
operations on $D$.
We call $\omega$ a \emph{weighted polymorphism}
of $\gamma$ if $C\subseteq\Pol(\gamma)$ and for any
$(\tup{x}_1,\ldots,\tup{x}_k) \in (\Feas(\gamma))^k$,
we have
\begin{equation}\label{eq:WPOL}
  \sum_{f \in C^{(k)}}\omega(f)\cdot\gamma(f(\tup{x}_1,\ldots,\tup{x}_k)) \ \leq\ 0\,.
\end{equation}
If $\omega$ is a weighted polymorphism of $\gamma$, we say that $\gamma$ is \emph{improved} by $\omega$.
\end{definition}
\begin{example}
Let $D = \{0, 1\}$ with ordering $0 < 1$. Binary operations $\min$ and $\max$
return the smaller and larger of their two arguments respectively. A function
$\gamma : D^k \to \q$ is \emph{submodular} if it satisfies $\gamma(\tup{x}_1) +
\gamma(\tup{x}_2) \geq \gamma(\min(\tup{x}_1, \tup{x}_2)) +
\gamma(\max(\tup{x}_1, \tup{x}_2))$ for all $\tup{x}_1, \tup{x}_2$. Clearly,
submodular functions are improved by the binary weighting $\omega = -\e{1}{2}
-\e{2}{2} + \min + \max$.
\end{example}

\begin{definition}
For any $\Gamma \subseteq \wrel_D$, we define $\wPol(\Gamma)$ to be the set of all
weightings of $\Pol(\Gamma)$ which are weighted polymorphisms of all weighted
relations $\gamma \in \Gamma$. We write $\wPol(\gamma)$ for $\wPol(\{\gamma\})$.
\end{definition}

\begin{definition}
We denote by $\wops_C$ the set of all possible (proper) weightings of clone $C$, and
define $\wops_D$ to be the union of the sets $\wops_C$ over all clones $C$ on $D$.
\end{definition}

Any $\Omega \subseteq \wops_D$ may contain weightings of \emph{different}
clones over $D$. We can then  extend each of these weightings with zeros, as
necessary, so that they are weightings of the same clone $C$, where $C$ is the
smallest clone containing all the clones associated with weightings in $\Omega$.

\begin{definition}
We define $\wClone(\Omega)$ to be the smallest weighted clone
containing this set of extended weightings obtained from $\Omega$.
\end{definition}

For any $\Omega \subseteq \wops_D$,
we denote by $\Imp(\Omega)$ the set of all weighted relations in $\wrel_D$ which are
improved by all weightings $\omega \in \Omega$.

The main result in~\cite{cccjz13:sicomp} establishes a 1-to-1
correspondence between weighted relational clones and weighted clones.

\begin{theorem}[
\cite{cccjz13:sicomp}]~\label{thm:wgc}
\begin{enumerate}
\item
For any finite $D$ and any finite $\Gamma \subseteq \wrel_D$, $\Imp(\wPol(\Gamma)) = \wRelClone(\Gamma)$.
\item
For any finite $D$ and any finite $\Omega\subseteq \wops_D$, $\wPol(\Imp(\Omega)) = \wClone(\Omega)$.
\end{enumerate}
\end{theorem}
Thus, when trying to identify tractable constraint languages, it is
sufficient to consider only languages of the form $\Imp(\Omega)$ for some weighted clone $\Omega$.

\section{Results}
\label{sec:results}

First we show that Theorem~\ref{thm:wgc} can be slightly extended to certain
constraint languages and sets of weightings of infinite size.

\begin{theorem}\hfill\label{thm:extGC}
\begin{enumerate}
\item Let $\Gamma \subseteq \wrel_D$. Then $\Imp(\wPol(\Gamma)) =
\wRelClone(\Gamma)$ if and only if $\wRelClone(\Gamma) = \Imp(\Omega)$ for some
$\Omega \subseteq \wops_D$.

\item
Let $\Omega \subseteq \wops_D$. Then $\wPol(\Imp(\Omega)) = \wClone(\Omega)$ if
and only if $\wClone(\Omega) = \wPol(\Gamma)$ for some $\Gamma \subseteq
\wrel_D$.
\end{enumerate}
\end{theorem}

\begin{proof}
We will only prove the first case as the second one is analogous.

Suppose that $\wRelClone(\Gamma) = \Imp(\Omega)$ for some $\Omega \subseteq
\wops_D$. As $\Gamma \subseteq \wRelClone(\Gamma)$, every weighting in $\Omega$
improves $\Gamma$, hence $\Omega \subseteq \wPol(\Gamma)$ and
$\Imp(\wPol(\Gamma)) \subseteq \Imp(\Omega) = \wRelClone(\Gamma)$. The inclusion
$\wRelClone(\Gamma) \subseteq \Imp(\wPol(\Gamma))$ follows from the fact that
$\Imp(\wPol(\Gamma))$ is a weighted relational
clone~\cite[Proposition~6.2]{cccjz13:sicomp} that contains $\Gamma$.

The converse implication holds trivially for $\Omega = \wPol(\Gamma)$.
\end{proof}

We remark that any \emph{finitely generated} weighted relational clone on a finite domain
satisfies, by Theorem~\ref{thm:wgc}\,(1), the condition of
Theorem~\ref{thm:extGC}\,(1). Similarly, any finitely generated weighted clone on a finite domain,
by Theorem~\ref{thm:wgc}\,(2), satisfies the condition of
Theorem~\ref{thm:extGC}\,(2).

However, our next result shows that Theorem~\ref{thm:wgc} does \emph{not} hold
for all infinite constraint languages and infinite sets of weightings.

\begin{theorem}\hfill\label{thm:counterQ}
\begin{enumerate}
\item There is a finite $D$ and an infinite $\Gamma \subseteq \wrel_D$ with
$\Imp(\wPol(\Gamma)) \neq \wRelClone(\Gamma)$.
\item
There is a finite $D$ and an infinite $\Omega\subseteq \wops_D$ with
$\wPol(\Imp(\Omega)) \neq \wClone(\Omega)$.
\end{enumerate}
\end{theorem}

Our aim is to establish a Galois connection even for infinite sets of weighted
relations and weightings. As we demonstrate in the proof of \Cref{thm:counterQ},
this cannot be done when restricted to rational weights; hence we allow weighted
relations and weightings to assign \emph{real-valued} weights. To distinguish
them from their formerly defined rational-valued counterparts, we will use a
subscript/superscript $\r$.

We will show in \Cref{lmWPolIsClosedClone} that $\RwPol(\Gamma)$ is
topologically closed (in a natural topology defined later) for any set of
weighted relations $\Gamma$; analogously, in \Cref{lmImpIsClosedClone} we will
show that $\RImp(\Omega)$ is topologically closed for any set of weightings
$\Omega$. Therefore, our new definitions of weighted (relational) clones require
them to be topologically closed.

Inspired by weighted pp-definitions~\cite{Thapper10:thesis}, we extend the
notion of weighted relational clones: we require them to be closed also under
operator $\Opt$. This change is justified by \Cref{thm:OPT} in which we prove
that the inclusion of $\Opt$ preserves tractability. In order to retain the
one-to-one correspondence with weighted clones, we need to alter their
definition too: weightings now assign weights to all operations and hence are
independent of the support clone (which becomes meaningless and we discard it).

Including the $\Opt$ operator brings two advantages to the study of weighted clones.
Firstly, it slightly simplifies the structure of the space of all weighted
clones. According to the original definition, a weighted clone is determined by
its support clone and the set of weightings it consists of; by our definition a
weighted clone equals the set of its weightings. Secondly, any non-projection
polymorphism of a weighted relational clone $\Gamma$ is assigned a positive
weight by some weighted polymorphism of $\Gamma$ (see
\Cref{thSupportWPolEqPol}).

Our main result is the following theorem, which holds for our new definition of
real-valued weightings and weighted relations.

\begin{theorem}[
Main]~\label{thm:main}
\begin{enumerate}
\item
For any finite $D$ and any $\Gamma \subseteq \Rwrel_D$, $\RImp(\RwPol(\Gamma)) =
\RwRelClone(\Gamma)$.
\item
For any finite $D$ and any $\Omega \subseteq \Rwops_D$, $\RwPol(\RImp(\Omega)) =
\RwClone(\Omega)$.
\end{enumerate}
\end{theorem}

Finally, we show that taking the weighted relational clone of a constraint
language preserves solvability with an absolute error bounded by $\epsilon$ (for
any $\epsilon > 0$), and demonstrate certain difficulties with proving that it
preserves exact solvability.

\section{Proof of Theorem~\ref{thm:counterQ}}
\label{sec:counterQ}

In this section we will prove Theorem~\ref{thm:counterQ}, which we state here as
two lemmas.

\begin{lemma}
\label{lmGammaCounterQ}
There is a finite $D$ and an infinite $\Gamma \subseteq \wrel_D$ with
$\Imp(\wPol(\Gamma)) \neq \wRelClone(\Gamma)$.
\end{lemma}

\begin{proof}
We set the domain to be $D = \{0, 1, 2\}$ and choose a positive \emph{irrational}
number $t$. Let $U \subseteq \wrel_D^{(1)}$ be the set of unary weighted
relations $\rho$ such that
\begin{equation}
\rho(2)-\rho(0) \geq (1+t) \cdot (\rho(1)-\rho(0))
\end{equation}
holds whenever $\rho(0), \rho(1), \rho(2)$ are all finite. It is easy to show
that $U$ is closed under addition, scaling by non-negative rational constants,
and addition of rational constants.

For any rational $u < t$, we define a unary weighted relation $\mu_u^- \in U$
such that $\mu_u^-(0) = 0$, $\mu_u^-(1) = -1$, and $\mu_u^-(2) = -1-u$. For any
rational $v > t$, we define a unary weighted relation $\mu_v^+ \in U$ such that
$\mu_v^+(0) = 0$, $\mu_v^+(1) = 1$, and $\mu_v^+(2) = 1+v$. It is easy to verify
that these weighted relations belong to $U$. Set $U$ also contains any unary
(unweighted) relation.

Let us define $\Gamma \subseteq \wrel_D$ as the set of weighted relations
$\gamma$ that can be written as
\begin{equation}
\gamma(x_1, \dots, x_r) =
\sum_{i=1}^r \rho_i(x_i) + \sum_{(i,j)\in S} \phi_=(x_i,x_j) \,,
\end{equation}
where $r$ equals the arity of $\gamma$, $\rho_i \in U$ for all $i$, $\phi_=$ is
the binary equality relation, and $S$ is an equivalence relation on $\{1, \dots,
r\}$. We claim that $\Gamma$ is a weighted relational clone. It certainly
contains $\phi_=$ and $\phi_\emptyset$, and is closed under addition, scaling by non-negative
rational constants, and addition of rational constants (as set $U$ is closed
under these operations). It is also closed under minimisation. Without loss of
generality, let us assume we minimise an $r$-ary weighted relation $\gamma$ ($r \geq 2$) over the
last variable ($x_r$). If the equivalence class of $r$ in $S$ is a singleton, we
simply add the value of $\min_{x_r\in D} \rho_r(x_r)$ to (say) $\rho_1$.
Otherwise, we can pick any $i \neq r$ such that $(i, r) \in S$ and replace
weighted relation $\rho_i$ with $\rho_i + \rho_r$.

We want to determine which weightings improve $\Gamma$. Let $\omega \in
\wPol^{(k)}(\Gamma)$ be a $k$-ary weighting and $\tup{x} \in D^k$. For any $a
\in D$, we will denote by $s_a$ the sum of weights $\omega(f)$ of all operations
$f$ such that $f(\tup{x}) = a$. Note that $s_0 + s_1 + s_2 = 0$. For any
rational $v > t$, weighting $\omega$ improves $\mu_v^+ \in \Gamma$, so we get
$s_1 + (1+v) \cdot s_2 \leq 0$ and therefore $v \cdot s_2 \leq s_0$. Similarly,
for any rational $u < t$, weighting $\omega$ improves $\mu_u^- \in \Gamma$, and
therefore $s_0 \leq u \cdot s_2$. We can choose both $u$ and $v$ arbitrarily
close to $t$, so it must hold $s_0 = t \cdot s_2$. However, $s_0$ and $s_2$ are
rational while $t$ is not. Therefore, we must have $s_0 = s_1 = s_2 = 0$ for any
weighting $\omega \in \wPol^{(k)}(\Gamma)$ and any $\tup{x} \in D^k$.

Now, let us consider the unary weighted relation $\rho$ defined as $\rho(0) = 0$
and $\rho(1) = \rho(2) = 1$. It follows from the previous paragraph that any
weighting $\omega \in \wPol(\Gamma)$ improves $\rho$, i.e.\ $\rho
\in \Imp(\wPol(\Gamma))$. However, $\rho \not\in \Gamma = \wRelClone(\Gamma)$,
so we get $\Imp(\wPol(\Gamma)) \neq \wRelClone(\Gamma)$.
\end{proof}

\begin{lemma}
\label{lmOmegaCounterQ}
There is a finite $D$ and an infinite $\Omega \subseteq \wops_D$ with
$\wPol(\Imp(\Omega)) \neq \wClone(\Omega)$.
\end{lemma}

\begin{proof}
We set the domain to be $D = \{0, 1, 2\}$ and choose a positive \emph{irrational}
number $t$. Let $C$ be the set of all operations $f$ such that $f(0,\dots,0) =
0$.  Clearly, $C$ contains all projections and is closed under superposition;
hence it is a clone. Let us define a set of weightings $\Omega \subseteq
\wops_D$ of the support clone $C$. For any arity $k \geq 1$, $\Omega^{(k)}$
consists of weightings $\omega$ such that for all $\mathbf{x} \in D^k$,
\begin{equation}
t\cdot\sum_{f(\mathbf{x})=2} \omega(f) \leq \sum_{f(\mathbf{x})=0} \omega(f)\,.
\label{eqRationalOmegaDefinition}
\end{equation}
It is easy to check that $\Omega$ is closed
under addition of weightings and non-negative scaling. To show that it is also
closed under superposition, let us consider any sequence of $\ell$-ary
operations $g_1, \dots, g_k$ and $\mathbf{x} \in D^\ell$. For any $a
\in D$ we have
\begin{equation}
\sum_{f(\mathbf{x})=a} \omega[g_1,\dots,g_k](f) =
\sum_{f[g_1,\dots,g_k](\mathbf{x})=a} \omega(f) =
\sum_{f(\mathbf{y})=a} \omega(f) \,,
\end{equation}
where $\mathbf{y} = (g_1(\mathbf{x}), \dots, g_k(\mathbf{x}))$. As $\omega$
satisfies Inequality~\eqref{eqRationalOmegaDefinition} for vector $\mathbf{y}$,
the superposition $\omega[g_1,\dots,g_k]$ satisfies it for vector $\mathbf{x}$.
Therefore, $\Omega$ is a weighted clone.

Let us denote by $c_0$ the unary constant zero operation, by $f,g$ the unary
operations and by $h$ the binary operation such that
\begin{equation}
\begin{aligned}
f(x) =
\begin{cases}
  0 & \text{for } x = 0 \\
  0 & \text{for } x = 1 \\
  2 & \text{for } x = 2
\end{cases}
\end{aligned} \,,
\qquad\qquad
\begin{aligned}
g(x) =
\begin{cases}
  0 & \text{for } x = 0 \\
  2 & \text{for } x = 1 \\
  2 & \text{for } x = 2
\end{cases}
\end{aligned} \,,
\qquad\qquad
\begin{aligned}
h(x,y) =
\begin{cases}
  1 & \text{for } x = 0 \wedge y = 2 \\
  2 & \text{for } x = 2 \wedge y = 2 \\
  0 & \text{otherwise}
\end{cases}
\end{aligned} \,.
\end{equation}
We denote by $\omega_0$ the unary weighting $-\e{1}{1} + c_0$. For any rational
$v > t$, we define a unary weighting $\mu_v^{(1)} = -(1+v) \cdot \e{1}{1} + v
\cdot f + g$. For any positive rational $u < t$, we define a binary weighting
$\mu_u^{(2)} = -u \cdot \e{1}{2} - \e{2}{2} + (1+u) \cdot h$. It is easy to show
that all these weightings belong to $\Omega$.

We will show that all weighted relations improved by $\Omega$ are relations,
i.e.\ $\Imp(\Omega) \subseteq \rel_D$. Suppose, to the contrary, that there is
an $r$-ary weighted relation $\gamma \in \Imp(\Omega)$ that is not a relation.
First, we obtain from it a ternary weighted relation with the same property.
Weighting $\omega_0$ improves $\gamma$, so we have $\gamma(\tup{0}) \leq
\gamma(\tup{x})$ for all $\tup{x} \in \Feas(\gamma)$, where $\tup{0} =
c_0(\tup{x})$ is the zero $r$-tuple. As $\gamma$ is not a relation, there must
be an $r$-tuple $\tup{z} = (z_1, \dots, z_r) \in \Feas(\gamma)$ for which
$\gamma(\tup{0}) < \gamma(\tup{z})$. Let us define a ternary weighted relation
$\rho$ so that $\rho(x_0, x_1, x_2) = \gamma(x_{z_1}, \dots, x_{z_r})$. It holds
that $\rho(0,0,0) = \gamma(\tup{0})$ and $\rho(0,1,2) = \gamma(\tup{z})$, so
$\rho(0,0,0) < \rho(0,1,2) < \infty$. Moreover, $\rho \in \Imp(\Omega)$.

For any rational $v > t$, weighting $\mu_v^{(1)}$ improves $\rho$. As $(0,1,2)
\in \Feas(\rho)$, we also have $(f(0),f(1),f(2)) = (0,0,2) \in \Feas(\rho)$,
$(g(0),g(1),g(2)) = (0,2,2) \in \Feas(\rho)$, and the inequality
\begin{equation}
\rho(0,2,2) - \rho(0,1,2) \leq v \cdot (\rho(0,1,2) - \rho(0,0,2)) \,.
\end{equation}
For any positive rational $u < t$, weighting $\mu_u^{(2)}$ improves $\rho$. As
$(0,0,2), (0,2,2) \in \Feas(\rho)$, we get
\begin{equation}
\rho(0,2,2) - \rho(0,1,2) \geq u \cdot (\rho(0,1,2) -\rho(0,0,2)) \,.
\end{equation}
We can choose both $u$ and $v$ arbitrarily close to $t$, so it must hold
\begin{equation}
\rho(0,2,2) - \rho(0,1,2) = t \cdot (\rho(0,1,2) - \rho(0,0,2)) \,.
\end{equation}
However, weights assigned by $\rho$ are rational while $t$ is not. Therefore,
$\rho(0,0,2) = \rho(0,1,2) = \rho(0,2,2)$. Similarly, by applying weightings
$\mu_u^{(2)}$ to $(0,0,0), (0,0,2)$ and weightings $\mu_v^{(1)}$ to $(0,0,1)$ we
obtain $\rho(0,0,0) = \rho(0,0,1) = \rho(0,0,2)$, which contradicts $\rho(0,0,0)
< \rho(0,1,2)$. Therefore, $\Imp(\Omega)$ contains only (unweighted) relations.

Now, let us consider the unary weighting $\omega = -\e{1}{1} + g$. Although it
does not belong to $\Omega = \wClone(\Omega)$ (it violates Inequality
\eqref{eqRationalOmegaDefinition} for $\tup{x} = (1)$), $\omega$ certainly
improves any $\gamma \in \Imp(\Omega)$. Therefore, $\wPol(\Imp(\Omega)) \neq
\wClone(\Omega)$.
\end{proof}

\section{New Galois Connection}
\label{sec:newGC}

In this section we will prove our main results. In Section~\ref{sub:prelim}, we
will describe the differences between the previous definitions of weighted
(relational) clones (as they were defined in~\cite{cccjz13:sicomp} and presented
in Section~\ref{sec:prelim} and the first part of Section~\ref{sec:results}) and
our new definitions. Section~\ref{sub:main} proves the main result, which
establishes a 1-to-1 correspondence between weighted relational clones and
weighted clones. Finally, Section~\ref{sub:complexity} is devoted to
computational-complexity consequences of our results.

\subsection{Preliminaries}
\label{sub:prelim}

Let $\rr=\mathbb{R}\cup\{\infty\}$ denote the set of real numbers with
(positive) infinity. We will allow weights in relations and weighted relations, as
defined in Definition~\ref{def:rel} and~\ref{def:wrel} respectively, to be real
numbers. In other words, an $m$-ary weighted relation $\gamma$ on $D$ is a
mapping $\gamma:D^m\to\rr$. We will add a subscript/superscript $\r$ to the
notation introduced in \Cref{sec:prelim} in order to emphasise the use of real
weights.

For any fixed arity $m$ and any $F \subseteq D^m$, consider the set of all
$m$-ary weighted relations $\gamma \in \Rwrel_D$ with $\Feas(\gamma) = F$. Let
us denote this set by $H$ and equip it with the inner product defined as
\begin{equation}
\inn{\alpha}{\beta} = \sum_{\tup{x}\in F} \alpha(\tup{x})\cdot\beta(\tup{x})
\end{equation}
for any $\alpha, \beta \in H$; $H$ is then a real Hilbert space. Set $\Rwrel_D$
is a disjoint union of such Hilbert spaces for all $m$ and $F$, and therefore a
topological space with the disjoint union topology induced by inner products on
the underlying Hilbert spaces. When we say a set of weighted relations is
open/closed, we will be referring to this topology.

\begin{definition}
\label{defn:wrelcloneOPT}
A constraint language $\Gamma \subseteq \Rwrel_D$ is called a \emph{weighted
relational clone} if it contains the binary equality relation $\phi_=$ and the
unary empty relation $\phi_\emptyset$, is closed under addition, minimisation,
scaling by non-negative real constants, addition of real constants, and operator
$\Opt$, and is topologically closed.

For any $\Gamma$, we define $\RwRelClone(\Gamma)$ to be the smallest weighted
relational clone containing $\Gamma$.
\end{definition}

As opposed to \Cref{defn:wrelclone}, our new definition requires weighted
relational clones to be closed under operator $\Opt$. In order to establish a
Galois connection now, we need to make an adjustment to the definition of
weighted clone too. We will discard the explicit underlying support clone;
instead, ($k$-ary) weightings will assign weights to all ($k$-ary) operations.
The role of the support clone of a weighted clone $\Omega$ is then taken over by
$\supp(\Omega)$ (see \Cref{lmSupportEqPolImp}).

We denote by $\ops_D^{(k)}$ the set of all $k$-ary operations on $D$ and let
$\ops_D=\bigcup_{k\geq 0} \ops_D^{(k)}$.

\begin{definition} \label{defn:wopNEW}
A $k$-ary \emph{weighting}
is a function $\omega : \ops_D^{(k)} \rightarrow \mathbb{R}$ such that
$\omega(f) < 0$ only if $f$ is a projection and
\begin{equation}\label{eqRZeroWeightsSum}
  \sum_{f \in \ops_D^{(k)}}\omega(f)\ =\ 0\,.
\end{equation}
We define $\supp(\omega)$ as
\begin{equation}
\supp(\omega) =
\JD^{(k)} \cup \left\{ f\in \ops_D^{(k)} \:\middle|\: \omega(f)>0 \right\} \,.
\end{equation}

We will call a function $\omega : \ops_D^{(k)} \rightarrow \mathbb{R}$ that
satisfies \Cref{eqRZeroWeightsSum} but assigns a negative weight to some
operation $f \not\in \JD^{(k)}$ an \emph{improper weighting}. In order to
emphasise the distinction we may also call a weighting a \emph{proper
weighting}.
\end{definition}

We denote by $\Rwops_D$ the set of all weightings on domain $D$. For any fixed
arity $k$, consider the set $H$ of all functions $\ops_D^{(k)} \to \r$
equipped with the
inner product defined as
\begin{equation}\label{eqWopsInnerProduct}
\inn{\alpha}{\beta} = \sum_{f\in\ops_D^{(k)}} \alpha(f)\cdot\beta(f)
\end{equation}
for any $\alpha, \beta \in H$; $H$ is then a real Hilbert space. Set $\Rwops_D$
lies in the disjoint union of such Hilbert spaces for all $k$, which is a
topological space with the disjoint union topology induced by inner products on
the underlying Hilbert spaces. When we say a set of weightings is open/closed,
we will be referring to this topology. Clearly, any closure point of a set of
weightings is itself a weighting.

\begin{definition} \label{defn:wcloneNEW}
Let $\Omega$ be a non-empty set of weightings on a fixed domain $D$. We define
$\supp(\Omega)=\JD\cup\bigcup_{\omega\in \Omega}\supp(\omega)$.

We call $\Omega$ a \emph{weighted clone} if it is closed under scaling by
non-negative real constants, addition of weightings of equal arity, and proper
superposition with operations from $\supp(\Omega)$, and is topologically closed.
\end{definition}

It is often convenient to build a desired proper weighting by taking a sum of
(possibly) improper superpositions. The following lemma, which is an analogue
of~\cite[Lemma~6.4]{cccjz13:sicomp}, shows that weighted
clones are closed under such constructions.

\begin{lemma}
\label{lmProperSumImproperWeightings}
Let $\Omega$ be a weighted clone, $\omega_1, \dots, \omega_n \in \Omega$, and
$c_1, \dots, c_n \geq 0$. We will denote the arity of weighting $\omega_i$ by
$\ell_i$. For any $1 \leq i \leq n$ and $1 \leq j \leq \ell_i$, let $g_{i,j} \in
\supp(\Omega)$ be a $k$-ary operation (for some fixed arity $k$). If the $k$-ary
weighting $\mu$ defined as
\begin{equation}
\mu = \sum_{i=1}^n c_i\cdot\omega_i[g_{i,1},\dots,g_{i,\ell_i}]
\end{equation}
is proper, then $\mu \in \Omega$.
\end{lemma}

\begin{proof}
We show that weighting $\mu$ can be constructed using proper superpositions
only.

Let us denote $\sum_{1\leq i \leq n} \ell_i$ by $t$. For any $1 \leq m \leq n$,
let $s_m = \sum_{1\leq i < m} \ell_i$.
A superposition with projections is always proper (as all negative weights are
transferred to projections), and therefore the $t$-ary weighting $\mu'$ defined
as
\begin{equation}
\mu' = \sum_{i=1}^n c_i\cdot\omega_i\left[\e{s_i+1}{t},\dots,\e{s_i+\ell_i}{t}\right]
\end{equation}
belongs to $\Omega$. Since $\mu =
\mu'[g_{1,1},\dots,g_{1,\ell_1},g_{2,1},\dots,g_{n,\ell_n}]$, we get $\mu \in
\Omega$.
\end{proof}

The following lemma has also been observed in~\cite{Kozik15:icalp,tz15:icalp}.

\begin{lemma}\label{lmSupportWCloneIsClone}
Let $\Omega$ be a weighted clone. Then $\supp(\Omega)$ is a clone.
\end{lemma}

\begin{proof}
We will denote $\supp(\Omega)$ by $C$. Since it contains all projections, we
only need to show that it is closed under superposition.

Let $f \in C^{(k)}$ and $g_1, \dots, g_k \in C^{(\ell)}$. If $f[g_1, \dots,
g_k]$ is a projection or is equal to $g_i$ for some $i$, then it clearly belongs
to $C$. Otherwise, $f$ is not a projection and therefore there is a $k$-ary
weighting $\omega \in \Omega$ for which $\omega(f) > 0$. Weighting $\omega[g_1,
\dots, g_k]$ certainly assigns a positive weight to $f[g_1, \dots, g_k]$ (we are
using the fact that only operations $g_1, \dots, g_k$ may receive negative
weight from projections in $\omega$). However, it is possibly improper, as it
may assign a negative weight to some $g_i$ that is not a projection.

We denote by $G$ the set of such operations $g \in \{g_1, \dots, g_k\}$ that are
not projections and $\omega[g_1, \dots, g_k](g) < 0$. For any $g \in G$, there
is an $\ell$-ary weighting $\mu_g \in \Omega$ for which $\mu_g(g) > 0$. Then the
$\ell$-ary weighting defined as
\begin{align}
\omega[g_1,\dots,g_k] + \sum_{g\in G}
\frac{-\omega[g_1,\dots,g_k](g)}{\mu_g(g)}\cdot\mu_g
\end{align}
is proper, belongs to $\Omega$ (by \Cref{lmProperSumImproperWeightings}), and
assigns a positive weight to $f[g_1, \dots, g_k]$.
\end{proof}

Again, we link weightings and weighted relations by the concept of weighted
polymorphism.

\begin{definition} \label{def:wpNEW}
Let $\gamma$ be an $m$-ary weighted relation
on $D$ and let $\omega$ be a $k$-ary weighting on $D$.
We call $\omega$ a \emph{weighted polymorphism}
of $\gamma$ if $\supp(\omega)\subseteq\Pol(\gamma)$ and for any
$(\tup{x}_1,\ldots,\tup{x}_k) \in (\Feas(\gamma))^k$,
we have
\begin{equation}\label{eq:wpolNEW}
  \sum_{f\in\supp(\omega)}\omega(f)\cdot\gamma(f(\tup{x}_1,\ldots,\tup{x}_k)) \ \leq\ 0\,.
\end{equation}
If $\omega$ is a weighted polymorphism of $\gamma$ we say that $\gamma$ is
\emph{improved} by $\omega$. We will denote the set of weighted polymorphisms of
$\Gamma$ by $\RwPol(\Gamma)$ and the set of weighted relations improved by
$\Omega$ by $\RImp(\Omega)$.
\end{definition}

The next lemma shows that $\supp(\Omega)$ consists of all polymorphisms of
$\RImp(\Omega)$ and hence fulfills the same role as the support clone in
\Cref{defn:wclone}.

\begin{lemma}\label{lmSupportEqPolImp}
Let $\Omega \subseteq \Rwops_D$ be a weighted clone. Then $\supp(\Omega) =
\Pol(\RImp(\Omega))$.
\end{lemma}

\begin{proof}
We will denote $\supp(\Omega)$ by $C$. Projections are polymorphisms of every
weighted relation, and any operation $f$ with $\omega(f) > 0$ for some $\omega
\in \Omega$ is a polymorphism of $\RImp(\Omega)$ by the definition of weighted
polymorphism. Therefore, $C \subseteq \Pol(\RImp(\Omega))$.

Let $\Inv(C)$ be the set of (unweighted) relations over $D$ that are invariant
under all operations from $C$ (i.e.\ operations from $C$ are their
polymorphisms). As any relation invariant under $\supp(\omega)$ is also improved
by $\omega$, we have $\Inv(C) \subseteq \RImp(\Omega)$ and thus
$\Pol(\RImp(\Omega)) \subseteq \Pol(\Inv(C)) = C$ (the last equality follows
from the Galois connection between relational clones and clones of
operations~\cite{Bodnarchuk69:closed,Geiger68:closed} and
\Cref{lmSupportWCloneIsClone}).
\end{proof}

The following corollary has been observed in the context of Min-Sol-Hom and
Min-Cost-Hom~\cite{Uppman13:icalp} by Hannes Uppman.\footnote{Private
communication, 2014.}

\begin{corollary}\label{thSupportWPolEqPol}
Let $\Gamma \subseteq \Rwrel_D$ be a weighted relational clone. Then
$\supp(\RwPol(\Gamma)) = \Pol(\Gamma)$.
\end{corollary}

\begin{proof}
We are going to use the Galois connection established later in \Cref{sub:main}.

As $\RwPol(\Gamma)$ is a weighted clone (\Cref{lmWPolIsClosedClone}), by
\Cref{lmSupportEqPolImp} we have $\supp(\RwPol(\Gamma)) =
\Pol(\RImp(\RwPol(\Gamma))) = \Pol(\Gamma)$ (the last equality follows from
\Cref{thGaloisGamma}).
\end{proof}

Finally, we introduce some notation that will be used throughout
\Cref{sub:main}. A sequence of $k$ $m$-tuples over $D$ can be written as $X =
(\tup{x}_1, \dots, \tup{x}_k) \in (D^m)^k$. By $X^T$ we will denote the
transpose of $X$, i.e.\ the sequence of $m$ $k$-tuples $(\tup{y}_1, \dots,
\tup{y}_m) \in (D^k)^m$ such that $\tup{y}_i = (\tup{x}_1[i], \dots,
\tup{x}_k[i])$. Let $f$ be a $k$-ary operation; we denote by $f(X)$ the
$m$-tuple obtained by applying $f$ coordinatewise to $\tup{x}_1, \dots,
\tup{x}_k$, i.e. $f(X) = f(\tup{x}_1, \dots, \tup{x}_k) = (f(\tup{y}_1), \dots,
f(\tup{y}_m))$.

Let $\gamma \in \Rwrel_D$ be a weighted relation and $\omega \in \Rwops_D$ a
$k$-ary weighting with $\supp(\omega) \subseteq \Pol(\gamma)$. Let us denote by
$H$ the Hilbert space of functions $\Pol^{(k)}(\gamma) \to \r$ with the inner
product analogous to \eqref{eqWopsInnerProduct}. As weighting $\omega$ assigns
non-zero weights only to operations from $\supp(\omega) \subseteq
\Pol^{(k)}(\gamma)$, we can identify $\omega$ with its restriction to
$\Pol^{(k)}(\gamma)$. For any $X \in (\Feas(\gamma))^k$, we denote by
$\gamma[X]$ the vector in $H$ such that $\gamma[X](f) = \gamma(f(X))$ for all $f
\in \Pol^{(k)}(\gamma)$. Inequality \eqref{eq:wpolNEW} can be then written as
$\inn{\omega}{\gamma[X]} \leq 0$.

The (internal) polar cone $K^\circ$ of a set $K \subseteq H$ is defined as
\begin{equation}
K^\circ = \Big\{ \alpha\in H ~\Big|~
             \inn{\alpha}{\beta} \leq 0 \text{ for all } \beta\in K \Big\} \,.
\end{equation}
It is well known (\cite{boyd2004convex,hiriart2001fundamentals}) that $K^\circ$
is a convex cone, i.e.\ $K^\circ$ is closed under addition of vectors and
scaling by non-negative constants. Moreover, $K^\circ$ is topologically closed,
and $K^{\circ\circ} = (K^\circ)^\circ$ is the closure of the smallest convex
cone containing $K$.%
\footnote{If $K$ is a finite set, the smallest convex cone containing $K$ is
topologically closed. This is why the former definitions of weighted
(relational) clones did not have to require topological closedness explicitly.
}

Let
\begin{equation}
K = \left\{ \gamma[X] ~\middle|~ X \in (\Feas(\gamma))^k \right\} \,.
\end{equation}
Weighting $\omega$ is then a weighted polymorphism of $\gamma$ if and only if
$\omega \in K^\circ$.

\subsection{Main Proofs}
\label{sub:main}

\begin{lemma}
\label{lmWPolIsClosedClone}
For any finite $D$ and any $\Gamma \subseteq \Rwrel_D$, $\RwPol(\Gamma)$ is a
weighted clone.
\end{lemma}

\begin{proof}
Let $k \geq 1$ be a fixed arity. We denote by $H$ the Hilbert space of functions
$\Pol^{(k)}(\Gamma) \to \mathbb{R}$ and define a set $K \subseteq H$ as
\begin{equation}
K = \left\{ \gamma[X] ~\middle|~
            \gamma\in\Gamma\wedge X\in(\Feas(\gamma))^k \right\} \,.
\end{equation}

A $k$-ary weighting $\omega$ with $\supp(\omega) \subseteq \Pol^{(k)}(\Gamma)$
is a weighted polymorphism of $\Gamma$ if and only if its restriction to
$\Pol^{(k)}(\Gamma)$ belongs to $K^\circ$. Set $\RwPol^{(k)}(\Gamma)$ is
therefore closed under addition as $K^\circ$ is convex, it is closed under
non-negative scaling as $K^\circ$ is a cone, and it is topologically closed as
$K^\circ$ is.

It remains to show that $\RwPol(\Gamma)$ is closed under superposition. Let
$\omega \in \RwPol(\Gamma)$ be a $k$-ary weighting and $g_1, \dots, g_k \in
\Pol(\Gamma)$ be $\ell$-ary operations. For any $\gamma \in \Gamma$ and $X \in
(\Feas(\gamma))^\ell$, we have $Y = (g_1(X), \dots, g_k(X)) \in
(\Feas(\gamma))^k$ and
\begin{align}
\sum_{f\in\supp(\omega[g_1,\dots,g_k])}
    \omega[g_1,\dots,g_k](f)\cdot\gamma(f(X))
&= \sum_{f\in\supp(\omega)} \omega(f)\cdot\gamma(f[g_1,\dots,g_k](X)) \\
&= \sum_{f\in\supp(\omega)} \omega(f)\cdot\gamma(f(Y)) \\
&\leq 0 \,.
\end{align}
Therefore, weighting $\omega[g_1, \dots, g_k]$ is a weighted polymorphism of
$\Gamma$.
\end{proof}

\begin{lemma}
\label{lmImpIsClosedClone}
For any finite $D$ and any $\Omega \subseteq \Rwops_D$, $\RImp(\Omega)$ is a
weighted relational clone.
\end{lemma}

\begin{proof}
Both $\phi_=$ and $\phi_\emptyset$ are improved by any weighting and hence
belong to $\RImp(\Omega)$. Addition, non-negative scaling, and addition of a
constant preserve Inequality~\eqref{eq:wpolNEW}, and therefore $\RImp(\Omega)$
is closed under these operations.

We need to prove that $\RImp(\Omega)$ is closed under minimisation. Let $\gamma
\in \RImp(\Omega)$ be an $r$-ary weighted relation and consider $\gamma'$
obtained from $\gamma$ by minimising over the last argument, i.e.
\begin{equation}
\gamma'(x_1, \dots, x_{r-1}) =
\min_{x_r \in D} \gamma(x_1, \dots, x_{r-1}, x_r) \,.
\end{equation}
Let $\omega \in \Omega$ be a $k$-ary weighting and $X' = (\mathbf{x}_1', \dots,
\mathbf{x}_k') \in (\Feas(\gamma'))^k$. For any $i \in \{1, \dots, k\}$, we can
extend $(r-1)$-tuple $\mathbf{x}_i'$ to an $r$-tuple $\mathbf{x}_i \in
\Feas(\gamma)$ so that $\gamma'(\mathbf{x}_i') = \gamma(\mathbf{x}_i)$; we will
denote the list of these extended $r$-tuples by $X = (\mathbf{x}_1, \dots,
\mathbf{x}_k)$. Note that $f(X)$ is an extension of $f(X')$ for any $k$-ary
operation $f \in \supp(\omega)$, and therefore $\gamma'(f(X')) \leq
\gamma(f(X))$. Moreover, $\gamma'(f(X')) = \gamma(f(X))$ whenever $f$ is a
projection. As $\gamma$ satisfies Inequality~\eqref{eq:wpolNEW} and only
projections may be assigned a negative weight, we have
\begin{equation}
\sum_{f\in\supp(\omega)} \omega(f)\cdot\gamma'(f(X')) \leq
\sum_{f\in\supp(\omega)} \omega(f)\cdot\gamma(f(X)) \leq
0 \,,
\end{equation}
and thus $\gamma' \in \RImp(\Omega)$.

Now we prove that $\RImp(\Omega)$ is closed under operator $\Opt$. Let $\gamma
\in \RImp(\Omega)$ and $\rho = \Opt(\gamma)$. We will assume that
$\Feas(\gamma)$ is non-empty (otherwise $\gamma = \rho$) and denote by $c$ the
minimum weight assigned by $\gamma$. Let $\omega \in \Omega$ be a $k$-ary
weighting. As $\rho$ is a relation, we only need to show that all operations in
the support of $\omega$ are polymorphisms of $\rho$. Let $X \in (\Feas(\rho))^k
\subseteq (\Feas(\gamma))^k$. For any operation $f$ in the support of $\omega$,
it holds $\gamma(f(X)) \geq c$. Moreover, $\gamma(f(X)) = c$ whenever $f$ is a
projection. We have
\begin{equation}
0 \geq \sum_{f\in\supp(\omega)} \omega(f)\cdot \gamma(f(X)) \geq
\sum_{f\in\supp(\omega)} \omega(f)\cdot c = 0 \,,
\end{equation}
which for all $f \in \supp(\omega)$ implies $\gamma(f(X)) = c$ and hence $f(X)
\in \Feas(\rho)$. Therefore, $\rho \in \RImp(\Omega)$.

Finally, we show that $\RImp(\Omega)$ is topologically closed. Let $r$ be a
fixed arity and $F \subseteq D^r$; we claim that the set $\Gamma \subseteq
\Rwrel_D$ of $r$-ary weighted relations $\gamma$ with $\Feas(\gamma) = F$ which
are \emph{not improved} by $\Omega$ is an open set. Take any $\gamma \in
\Gamma$. There must be a non-zero weighting $\omega \in \Omega$ (let us denote
its arity by $k$) and $X \in F^k$ such that $\langle \omega, \gamma[X] \rangle =
d$ for some positive $d$, i.e.\ $\omega$ violates Inequality~\eqref{eq:wpolNEW}
for $\gamma$ and $X$. Then for every $r$-ary weighted relation $\gamma'$ with
$\Feas(\gamma') = F$ and distance from $\gamma$ less than
\begin{equation}
\frac{d}{\displaystyle\sum_{f\in\supp(\omega)} |\omega(f)|}
\end{equation}
it holds $\langle \omega, \gamma'[X]
\rangle > 0$, so $\gamma$ has a neighbourhood contained in $\Gamma$.
\end{proof}

\noindent
We are now ready to prove our main result, stated as
Theorems~\ref{thGaloisGamma} and \ref{thGaloisOmega}.

\begin{theorem}
\label{thGaloisGamma}
For any finite $D$ and any $\Gamma \subseteq \Rwrel_D$, $\RImp(\RwPol(\Gamma)) =
\RwRelClone(\Gamma)$.
\end{theorem}

\begin{proof}
First, we show that $\RwRelClone(\Gamma) \subseteq \RImp(\RwPol(\Gamma))$.
Surely $\Gamma \subseteq \RImp(\RwPol(\Gamma))$ and therefore
$\RwRelClone(\Gamma) \subseteq \RwRelClone(\RImp(\RwPol(\Gamma)))$. By
\Cref{lmImpIsClosedClone} we know that $\RImp(\RwPol(\Gamma))$ is a weighted
relational clone. Hence, $\RwRelClone(\RImp(\RwPol(\Gamma))) =
\RImp(\RwPol(\Gamma))$.

Now we will prove the other inclusion, $\RImp(\RwPol(\Gamma)) \subseteq
\RwRelClone(\Gamma)$. Let $\rho \in \RImp(\RwPol(\Gamma))$ be a
weighted relation and denote $|\Feas(\rho)|$ by $k$. If $k=0$, $\rho$ is
expressible from $\phi_\emptyset$ and hence $\rho \in \RwRelClone(\Gamma)$.
Otherwise, we will focus solely on the $k$-ary weighted polymorphisms of
$\Gamma$. Let us denote the Hilbert space of functions $\Pol^{(k)}(\Gamma) \to
\mathbb{R}$ by $H$; the $k$-ary weighted polymorphisms of $\Gamma$ can be then
seen as vectors from $H$. Denoting $m = |D|^k$, a $k$-ary operation on $D$ is
uniquely determined by the $m$-tuple of labels it assigns to its $m$ possible
inputs. Later we will define a correspondence between a subset of $m$-ary
weighted relations and $H$.

The outline of the proof is as follows. We transform $\Gamma$ into a set $M
\subseteq \RwRelClone(\Gamma)$ of $m$-ary weighted relations and consider their
corresponding vectors in $H$. The polar cone of these vectors equals
$\RwPol^{(k)}(\Gamma)$, and its polar cone, in turn, consists of $m$-ary
weighted relations improved by $\RwPol^{(k)}(\Gamma)$. We know that the polar
cone of the polar cone of a set is the closure of the smallest convex cone
containing this set; therefore, any $m$-ary weighted relation improved by
$\RwPol^{(k)}(\Gamma)$ belongs to $\RwRelClone(\Gamma)$. This also includes a
particular $m$-ary weighted relation that we use to express $\rho$, so we get
$\rho \in \RwRelClone(\Gamma)$.

We begin by formally defining the correspondence between certain $m$-ary weighted
relations and vectors from $H$. Let us denote by $(\tup{z}_1, \dots, \tup{z}_m)
= Z^T$ the sequence of all $k$-tuples over $D$ in an arbitrary fixed order; any
$k$-ary operation $f$ is then determined by the $m$-tuple $(f(\tup{z}_1), \dots,
f(\tup{z}_m)) = f(Z)$. Let us define a set $F \subseteq D^m$ as
\begin{equation}
F = \left\{ f(Z) ~\middle|~ f \in \Pol^{(k)}(\Gamma) \right\} \,.
\end{equation}
For any $m$-ary weighted relation $\gamma$ with $\Feas(\gamma) = F$, the
corresponding vector in $H$ is $\gamma[Z]$. Conversely, for any vector in $H$
there is a corresponding $m$-ary weighted relation with finite weights precisely
on $F$.

Now we transform $\Gamma$ into a set of $m$-ary weighted relations $M \subseteq
\RwRelClone(\Gamma)$ that captures enough information to reconstruct the set of
$k$-ary weighted polymorphisms of $\Gamma$. Let $\gamma \in \Gamma$ be an
$n$-ary weighted relation and $X \in (\Feas(\gamma))^k$; we will denote the
$k$-tuples of $X^T$ by $(\tup{x}_1, \dots, \tup{x}_n)$. We claim that there is
an $m$-ary weighted relation $\mu_{\gamma,X} \in \RwRelClone(\Gamma)$ with
$\Feas(\mu_{\gamma,X}) = F$ such that $\mu_{\gamma,X}(f(Z)) = \gamma(f(X))$ for
all $f \in \Pol^{(k)}(\Gamma)$. First, we construct an $m$-ary weighted relation
$\mu_{\gamma,X}'$ as
\begin{equation}
\mu_{\gamma,X}'(y_{\tup{z}_1}, \dots, y_{\tup{z}_m}) =
\gamma(y_{\tup{x}_1}, \dots, y_{\tup{x}_n}) \,,
\end{equation}
where $y_{\tup{z}_i}$ are variables indexed by $k$-tuples over $D$. Clearly,
$\mu_{\gamma,X}'(f(Z)) = \mu_{\gamma,X}'(f(\tup{z}_1), \dots, f(\tup{z}_m)) =
\gamma(f(\tup{x}_1), \dots, f(\tup{x}_n)) = \gamma(f(X))$. However, we are not done yet,
as $\mu_{\gamma,X}'$ assigns a finite weight to all $m$-tuples $f(Z)$ such that
$f(X) \in \Feas(\gamma)$, even if $f \not\in \Pol^{(k)}(\Gamma)$. We can easily
fix this: Let $f$ be an $k$-ary operation that is not a polymorphism of
$\Gamma$; then there is a weighted relation $\gamma_f \in \Gamma$ and $X_f \in
(\Feas(\gamma_f))^k$ such that $f(X_f) \not\in \Feas(\gamma_f)$. Adding $0 \cdot
\mu_{\gamma_f,X_f}'$ to $\mu_{\gamma,X}'$ ensures that the weighted relation
will assign infinity to $m$-tuple $f(Z)$ without changing other weights. This
can be done for all (finitely many) such operations $f$, so we obtain a weighted
relation $\mu_{\gamma,X}$ with $\Feas(\mu_{\gamma,X}) = F$.

Similarly, there are $m$-ary weighted relations $\mu_\iota, \mu_{-\iota} \in
\RwRelClone(\Gamma)$ with $\Feas(\mu_\iota) = \Feas(\mu_{-\iota}) = F$ such that
$\mu_\iota(f(Z)) = 1$ and $\mu_{-\iota}(f(Z)) = -1$ for all $f \in
\Pol^{(k)}(\Gamma)$. Again, we start with $\mu_\iota'(y_{\tup{z}_1}, \dots,
y_{\tup{z}_m}) = 1$, $\mu_{-\iota}'(y_{\tup{z}_1}, \dots, y_{\tup{z}_m}) = -1$
and then add $0 \cdot \mu_{\gamma_f,X_f}'$ for all $k$-ary operations $f \not\in
\Pol^{(k)}(\Gamma)$ to ensure that the resulting weighted relations $\mu_\iota,
\mu_{-\iota}$ assign finite weights only to $m$-tuples from $F$.

Let $\iota \in H$ be the vector assigning every operation value $1$. For any
operation $f \in \Pol^{(k)}(\Gamma)$ that is not a projection, let $\eps_f \in
H$ be the vector such that $\eps_f(f) = 1$ and $\eps_f(g) = 0$ for all $g \neq
f$. We define a set of $m$-ary weighted relations $M \subseteq
\RwRelClone(\Gamma)$, the set of corresponding vectors $V \subseteq H$, and an
auxiliary set of vectors $W \subseteq H$ as follows:
\begin{align}
M &= \left\{ \mu_{\gamma,X} ~\middle|~
             \gamma\in\Gamma \wedge X\in(\Feas(\gamma))^k \right\}
\cup \left\{ \mu_\iota, \mu_{-\iota} \right\} \\
V &= \left\{ \gamma[X] ~\middle|~
             \gamma\in\Gamma \wedge X\in(\Feas(\gamma))^k \right\}
\cup \left\{ \iota, -\iota \right\} \\
W &= V\cup\left\{ -\eps_f ~\middle|~
                  f\in\Pol^{(k)}(\Gamma)\setminus\JD^{(k)} \right\} \,.
\end{align}
We claim that the polar cone $W^\circ$ consists of $k$-ary weighted
polymorphisms of $\Gamma$. Let $\omega \in W^\circ$ be a vector. As $\langle
\omega, \iota \rangle \leq 0$ and $\langle \omega, -\iota \rangle \leq 0$, we
have $\langle \omega, \iota \rangle = 0$, i.e.\ the sum of weights of $\omega$
equals $0$. For any non-projection $f$ we have $\langle \omega, -\eps_f \rangle
\leq 0$, i.e.\ $\omega(f)$ is non-negative. Finally, for any $\gamma \in \Gamma$
and $X \in (\Feas(\gamma))^k$ it holds $\langle \omega, \gamma[X] \rangle \leq
0$; hence $\omega$ is a weighted polymorphism of $\Gamma$.


Let us now return to weighted relation $\rho$ and denote the sequence of
elements of $\Feas(\rho)$ in an arbitrary fixed order by $R \in
(\Feas(\rho))^k$. As $\rho$ is improved by $\RwPol(\Gamma)$, any vector $\omega
\in W^\circ$ satisfies Inequality~\eqref{eq:wpolNEW} for $\rho$ and any $X \in
(\Feas(\rho))^k$, in particular for $X = R$. Hence, we would like to claim that
$\langle \omega, \rho[R] \rangle \leq 0$ for all $\omega \in W^\circ$ and thus
$\rho[R] \in W^{\circ\circ}$. However, $\rho[R]$ might be ill-defined: Although
$\rho \in \RImp(\RwPol(\Gamma))$, not necessarily all operations $f \in
\Pol^{(k)}(\Gamma)$ are polymorphisms of $\rho$, and therefore possibly
$\rho(f(R)) = \infty$. Let us denote the set of these problematic operations by
\begin{equation}
Q = \left\{ f \in \Pol^{(k)}(\Gamma) ~\middle|~
            f(R) \not\in \Feas(\rho) \right\} \,.
\end{equation}
On the other hand, every operation in the support of $\RwPol^{(k)}(\Gamma)$ is a
polymorphism of $\rho$. This implies that operations in $Q$ must be assigned a
zero weight by all $\omega \in W^\circ$. As $\rho[R]$ might not exist, let us
define instead a substitute vector $\beta \in H$ such that $\beta(f) =
\rho(f(R))$ for all $f \in \Pol^{(k)}(\Gamma) \setminus Q$, with arbitrary
values assigned to operations in $Q$. By the previous argument, $\beta \in
W^{\circ\circ}$. Additionally, let $\beta_0 \in H$ be a vector such that
$\beta_0(f) > 0$ if $f \in Q$, and $\beta_0(f) = 0$ otherwise. For any $\omega
\in W^\circ$ it holds $\langle \omega, \beta_0 \rangle = 0$, so $\beta_0$ also
belongs to $W^{\circ\circ}$.

Any vector in $W^{\circ\circ}$ can be obtained from some vector in
$V^{\circ\circ}$ by adding non-negative multiples of $-\eps_f$ for $f \in
\Pol^{(k)}(\Gamma) \setminus \JD^{(k)}$. Therefore, there is a vector $\alpha
\in V^{\circ\circ}$ such that $\alpha(f) \geq \beta(f) = \rho(f(R))$ for all $f
\not\in Q$, and $\alpha(f) = \beta(f) = \rho(f(R))$ when $f$ is a projection.
Also, there is a non-negative vector $\alpha_0 \in V^{\circ\circ}$ such that
$\alpha_0(f) \geq \beta_0(f) > 0$ if $f \in Q$, and $\alpha_0(f) = \beta_0(f) =
0$ if $f$ is a projection.

Vectors in $V$ correspond to weighted relations in $M \subseteq
\RwRelClone(\Gamma)$. Set $V^{\circ\circ}$ is the closure of the smallest convex
cone containing $V$, and therefore weighted relations corresponding to vectors
in $V^{\circ\circ}$ also belong to $\RwRelClone(\Gamma)$ (as it is closed under
addition and non-negative scaling, and is topologically closed). Hence, there
are $m$-ary weighted relations $\psi, \psi_0 \in \RwRelClone(\Gamma)$ with
$\Feas(\psi) = \Feas(\psi_0) = F$ such that $\psi[Z] = \alpha$ and $\psi_0[Z] =
\alpha_0$. We are going to express $\rho$ from them.

Let us denote the arity of $\rho$ by $n$ and the $k$-tuples of $R^T$ by
$(\tup{r}_1, \dots, \tup{r}_n)$. Consider the following gadget. Let
$I$ be a VCSP instance with variables $y_{\tup{z}_1}, \dots,
y_{\tup{z}_m}$ and a single constraint $\psi(y_{\tup{z}_1}, \dots,
y_{\tup{z}_m})$, and let $L = (y_{\tup{r}_1}, \dots, y_{\tup{r}_n})$. Then
$\pi_L(I)$ is an $n$-ary weighted relation expressible over
$\RwRelClone(\Gamma)$; we will denote it by $\rho'$. For any
$n$-tuple $\tup{x} \in D^n$, we have
\begin{align}
\rho'(\tup{x})
&= \min_{\{ (y_{\tup{z}_1}, \dots, y_{\tup{z}_m}) \in D^m ~|~
            (y_{\tup{r}_1}, \dots, y_{\tup{r}_n}) = \tup{x} \}}
     \psi(y_{\tup{z}_1}, \dots, y_{\tup{z}_m}) \\
&= \min_{\{ f : D^k \to D ~|~
            (f(\tup{r}_1), \dots, f(\tup{r}_n)) = \tup{x} \}}
     \psi(f(\tup{z}_1), \dots, f(\tup{z}_m)) \\
&= \min_{f(R) = \tup{x}} \psi(f(Z)) =
\min_{f(R) = \tup{x}} \alpha(f) \,.
\end{align}
Analogously, by replacing $\psi$ with $\psi_0$ in the gadget we obtain an
$n$-ary weighted relation $\rho_0'$ for which $\rho_0'(\tup{x}) =
\min_{f(R)=\tup{x}} \alpha_0(f)$.

For any $\tup{x} \in \Feas(\rho)$ and $k$-ary operation $f$ such that $f(R) =
\tup{x}$, it holds $\alpha(f) \geq \rho(f(R)) = \rho(\tup{x})$. As $R$ is
a list of all elements of $\Feas(\rho)$, there is a projection $f$ such that
$f(R) = \tup{x}$; for it we have $\alpha(f) = \rho(f(R)) = \rho(\tup{x})$.
Therefore, $\rho'(\tup{x}) = \rho(\tup{x})$. Similarly we get
$\rho_0'(\tup{x}) = 0$ for any $\tup{x} \in \Feas(\rho)$.

We are almost done; the last issue is that $\rho'(\tup{x})$ may be finite
also for some $\tup{x} \not\in \Feas(\rho)$. But $f(R) \not\in \Feas(\rho)$
implies $f \in Q$, and in that case $\alpha_0(f)$ is positive. Therefore,
$\Opt(\rho_0')$ is finite only on $\Feas(\rho)$, and $\rho' + \Opt(\rho_0') =
\rho$.
\end{proof}

\begin{theorem}
\label{thGaloisOmega}
For any finite $D$ and any $\Omega \subseteq \Rwops_D$, $\RwPol(\RImp(\Omega)) =
\RwClone(\Omega)$.
\end{theorem}

\begin{proof}
We begin with the inclusion $\RwClone(\Omega) \subseteq \RwPol(\RImp(\Omega))$.

Weightings in $\Omega$ are weighted polymorphisms of all weighted relations in
$\RImp(\Omega)$, so $\Omega \subseteq \RwPol(\RImp(\Omega))$, and hence
$\RwClone(\Omega) \subseteq \RwClone(\RwPol(\RImp(\Omega)))$. By
\Cref{lmWPolIsClosedClone}, we have that $\RwPol(\RImp(\Omega))$ is a weighted
clone, so $\RwClone(\RwPol(\RImp(\Omega))) = \RwPol(\RImp(\Omega))$.

Now we prove that for any $k \geq 1$ and any $k$-ary weighting $\mu \in
\RwPol(\RImp(\Omega))$, it holds $\mu \in \RwClone(\Omega)$. First,
let us establish the clone of operations we will be working with. Let $C$ be the
smallest clone containing $\supp(\Omega)$. The support of $\RwClone(\Omega)$ is
itself a clone (by \Cref{lmSupportWCloneIsClone}) so we also have $C =
\supp(\RwClone(\Omega))$. As in the proof of \Cref{thGaloisGamma}, we will
represent $k$-ary weightings by vectors of the Hilbert space $H = C^{(k)} \to
\mathbb{R}$, and identify those vectors with certain $m$-ary weighted relations (where $m
= |D|^k$).

The outline of the proof is as follows.
We transform $\Omega$ into a set $W$ of $k$-ary weightings. Although some of
these weightings may be improper, any proper weighting obtained as their
non-negative linear combination belongs to $\RwClone(\Omega)$. The polar cone
$W^\circ$ consists of $m$-ary weighted relations improved by $\Omega$, and its
polar cone $W^{\circ\circ}$ hence contains $\mu$. As the polar cone of the polar
cone of a set is the closure of the smallest convex cone containing this set, we
get $\mu \in \RwClone(\Omega)$.

Recall the correspondence between a subset of $m$-ary weighted relations and $H$
from the proof of \Cref{thGaloisGamma}. This time, we are working with clone
$C$, so we define $F$ as
\begin{equation}
F = \left\{ f(Z) ~\middle|~ f \in C^{(k)} \right\} \,.
\end{equation}
Let $\Gamma$ be the set of all $m$-ary weighted relations $\gamma$ with
$\Feas(\gamma) = F$. Similarly as before, there is a bijection between $\Gamma$
and $H$: the corresponding vector to a weighted relation $\gamma \in \Gamma$ is
$\gamma[Z]$.

We show that $k$-ary polymorphisms of any $\gamma \in \Gamma$ are precisely the
operations from $C^{(k)}$. Let $f \in C^{(k)}$. For any $X \in F^k$, there are
$k$-ary operations $g_1, \dots, g_k \in C^{(k)}$ such that $X = (g_1(Z), \dots,
g_k(Z))$. So we have $f(X) = f[g_1, \dots, g_k](Z) \in F$ because $f[g_1, \dots,
g_k] \in C^{(k)}$. Conversely, let $f$ be a $k$-ary operation not belonging to
$C^{(k)}$. Certainly $Z = (\e{1}{k}(Z), \dots, \e{k}{k}(Z)) \in F^k$, but $f(Z)
\not\in F$. Therefore, $f$ is not a polymorphism of $\gamma$.

Let us define a set $W \subseteq H$ as
\begin{equation}
W = \left\{ \omega[g_1,\dots,g_\ell] \in H ~\middle|~
            \ell\geq1 \wedge \omega\in\Omega^{(\ell)} \wedge
            g_1,\dots,g_\ell\in C^{(k)}
    \right\} \,.
\end{equation}
We claim that for any vector in the polar cone $W^\circ$, the corresponding
weighted relation is improved by $\Omega$. Let $\gamma \in \Gamma$ be a weighted
relation such that $\gamma[Z] \in W^\circ$, $\omega \in \Omega$ be an $\ell$-ary
weighting, and $X \in F^\ell$. Then there are $k$-ary operations $g_1, \dots,
g_\ell \in C^{(k)}$ for which $X = (g_1(Z), \dots, g_\ell(Z))$, and we have
\begin{align}
\sum_{f\in\supp(\omega)} \omega(f)\cdot\gamma(f(X))
&= \sum_{f\in C^{(\ell)}} \omega(f)\cdot\gamma(f[g_1,\dots,g_\ell](Z)) \\
&= \sum_{f\in C^{(k)}} \omega[g_1,\dots,g_\ell](f)\cdot\gamma(f(Z)) \\
&= \langle \omega[g_1,\dots,g_\ell], \gamma[Z] \rangle \leq 0 \,.
\end{align}

Weighting $\mu$ is a weighted polymorphism of $\RImp(\Omega)$, so it improves
any weighted relation $\gamma$ corresponding to a vector in $W^\circ$. Firstly,
this implies $\supp(\mu) \subseteq C^{(k)}$; we can therefore view $\mu$ as a
vector of $H$. Secondly, $\mu$ satisfies Inequality~\eqref{eq:wpolNEW}
for $\gamma$ and any $X \in F^k$. In particular, $Z \in F^k$, so we get $\langle
\mu, \gamma[Z] \rangle \leq 0$ and thus $\mu \in W^{\circ\circ}$.

Set $W^{\circ\circ}$ is the closure of the smallest convex cone containing $W$.
By \Cref{lmProperSumImproperWeightings}, any proper weighting obtained as a
non-negative linear combination of weightings from $W$ belongs to
$\RwClone(\Omega)$. Therefore, $\mu \in \RwClone(\Omega)$.
\end{proof}

\subsection{Complexity Consequences}
\label{sub:complexity}

When studying the computational complexity of constraint languages, the focus on
weighted relational clones is justified by \Cref{thm:relclonetract}. The aim of
this section is to discuss the consequences of our changes in the definition of
weighted relational clones (allowing real weights and requiring weighted
relational clones to be closed under operator $\Opt$ and to be topologically
closed) on the validity of that theorem. We will assume the arithmetic model of
computation, i.e.\ basic arithmetic operations with real numbers take a constant
time.

First we show that adding operator $\Opt$ preserves the tractability of weighted
relational clones.

\begin{theorem}\label{thm:OPT}
Let $\Gamma \subseteq \Rwrel_D$ be a finite constraint language and
$\gamma\in\Gamma$. Then $\VCSP(\Gamma\cup\{\Opt(\gamma)\})$ polynomial-time
reduces to $\VCSP(\Gamma)$.
\end{theorem}

\begin{proof}
Adding a constant to all weights of a weighted relation changes the value of
every assignment by the same amount, and hence does not affect tractability.
Without loss of generality, we may therefore assume that all weighted relations
in $\Gamma$ assign non-negative weights and that the minimum weight assigned by
$\gamma$ equals $0$. We will also assume that $\gamma$ is not a relation,
otherwise $\Opt(\gamma) = \gamma$ so the claim would hold trivially.  Let us
denote by $m$ the smallest positive weight assigned by $\gamma$, and by $M$ the
largest finite weight assigned by any $\gamma' \in \Gamma$.

Let $I \in \VCSP(\Gamma \cup \{ \Opt(\gamma) \})$ be an instance with
$q$ constraints. We replace every constraint of the form $\Opt(\gamma)(\tup{x})$
in $I$ with $(q \cdot \lceil M/m \rceil + 1)$ copies of
$\gamma(\tup{x})$, thus obtaining a polynomially larger instance $I'
\in \VCSP(\Gamma)$. Any feasible assignment for instance $I$ is also
feasible for $I'$ with the same value, which does not exceed $qM$. On the other
hand, any infeasible assignment for instance $I$ is
infeasible for $I'$, or it incurs an infinite value from a constraint
of the form $\Opt(\gamma)(\tup{x})$ in $I$ and therefore a value of
more than $qM$ in $I'$.
\end{proof}

For any constraint language $\Gamma \subseteq \Rwrel_D$, we will denote by
$\costeq{\Gamma}$ the smallest set of weighted relations containing $\Gamma$
that is closed under scaling by non-negative real constants and addition of real
constants. Analogously to~\cite[Theorem~4.3]{cccjz13:sicomp}, we would like to
show that $\Gamma$ is tractable if and only if $\costeq{\Gamma}$ is tractable.
Their proof, however, does not apply to scaling by an \emph{irrational} factor
$\alpha$, as it relies on the existence of integers $p,q$ such that $\alpha =
p/q$. In fact, we were not able to prove that real-valued scaling preserves
tractability when insisting on exact solvability. If we consider solving $\VCSP$
with an absolute error bounded by $\epsilon$ (for any $\epsilon > 0$), then
real-valued scaling \emph{does preserve} tractability, as shown in the following
theorem.

\begin{theorem}\label{thm:realScaling}
Let $\Gamma, \Gamma' \subseteq \Rwrel_D$ be finite constraint languages such
that $\Gamma$ contains only weighted relations of the form $c \cdot \gamma'$ for
$c \geq 0, \gamma' \in \Gamma'$. For any $\epsilon>0$, there is a
polynomial-time reduction that for any instance $I \in \VCSP(\Gamma)$ outputs an
instance $I' \in \VCSP(\Gamma')$ such that for any optimal assignment $s'$ of
$I'$ it holds $I(s') \in [v, v+\epsilon]$, where $v$ is the value of an optimal
assignment of $I$.
\end{theorem}

\begin{proof}
Again, we may assume that all weighted relations in $\Gamma$ and $\Gamma'$
assign non-negative weights. Let us denote by $M$ the largest finite weight
assigned by any weighted relation in $\Gamma'$; we may assume that $M$ is
well-defined and positive, otherwise we would have $\Gamma \subseteq \Gamma'$.
We will denote by $q$ the number of constraints in $I$ and let $b = \lceil
qM/\epsilon \rceil$.

Let instance $I' \in \VCSP(\Gamma')$ have the same set of variables as $I$. For
any constraint $c_i \cdot \gamma_i'(\tup{x}_i)$ of $I$, we add $(\lfloor bc_i
\rfloor + 1)$ copies of constraint $\gamma_i'(\tup{x}_i)$ into $I'$. Note that
any feasible assignment of $I$ is a feasible assignment of $I'$, and vice versa.
Let us assume that $I$ admits a feasible solution. As $(\lfloor bc_i \rfloor +
1) - bc_i \in (0,1]$, we have
\begin{equation}
b \cdot I(t) \leq I'(t) \leq b \cdot I(t) + qM
\end{equation}
for any feasible assignment $t$. Let $s$ be an optimal assignment of $I$; we get
\begin{equation}
b\cdot I(s') \leq I'(s') \leq I'(s) \leq b\cdot I(s)+qM \,,
\end{equation}
and therefore $I(s') \leq I(s) + \epsilon = v + \epsilon$.
\end{proof}

Taking a topological closure of a language also preserves tractability with a
bounded absolute error, as the following theorem shows.

\begin{theorem}\label{thm:closureTractability}
Let $\Gamma \subseteq \Rwrel_D$ be a constraint language and denote by
$\overline{\Gamma}$ its closure. For any $\epsilon > 0$, there is a
polynomial-time reduction that for any instance $I \in \VCSP(\overline{\Gamma})$
outputs an instance $I' \in \VCSP(\Gamma)$ with the same variables such that any
assignment $t$ is either infeasible for both $I$ and $I'$, or is feasible for
both and $|I(t) - I'(t)| \leq \epsilon$.
\end{theorem}

\begin{proof}
Let us denote by $q$ the number of constraints in $I$. For any $\gamma \in
\overline{\Gamma}$, there is a weighted relation $\gamma' \in \Gamma$ of the
same arity and with $\Feas(\gamma) = \Feas(\gamma')$ such that the distance
between $\gamma$ and $\gamma'$ is at most $\epsilon / q$. We obtain the sought
instance $I'$ by replacing all constraints $\gamma$ from $I$ with their
counterparts $\gamma'$.
\end{proof}

We finish this section with a discussion of the difficulty of improving
\Cref{thm:realScaling} to exact solvability (to optimality).
Let $\Gamma$ be a finite constraint language and $\gamma\in\Gamma$. We would like to
prove that $\VCSP(\Gamma\cup\{c\cdot\gamma\})$ polynomial-time reduces to
$\VCSP(\Gamma)$, where $c\in\r_{\geq 0}$. Given
$I\in\VCSP(\Gamma\cup\{c\cdot\gamma\})$, let
\begin{equation}
   \delta_I = \min \Big\{
     |I(s_1)-I(s_2)| ~\Big|~
     \text{$s_1, s_2$ are solutions to $I$ with different values}
   \Big\}\,.
\end{equation}
If we choose an $\epsilon < \delta_I$, we obtain an optimal solution of $I$ by
\Cref{thm:realScaling}. However, it is not clear how fast the value of
$\delta_I$ approaches $0$ as the size of $I$ grows to infinity, and whether it
is possible to compute it in polynomial time.

\section*{Acknowledgments}

The authors are grateful to the anonymous referees for their helpful comments
and suggestions.


\newcommand{\noopsort}[1]{}

\end{document}